\documentclass[12pt]{amsart}

\usepackage{amssymb}
\usepackage{amsmath}
\usepackage{latexsym}
\usepackage{graphicx}
\usepackage{color}

\normalsize \evensidemargin=0pt \oddsidemargin=0pt \hoffset=6pt
\theoremstyle{plain}

\newtheorem{theorem}{\bf Theorem}

\newtheorem{proposition}[theorem]{\bf Proposition}
\newtheorem{corollary}[theorem]{\bf Corollary}

\newtheorem{example}{Example}[section]

\def\bds{\begin{displaystyle}}
\def\eds{\end{displaystyle}}

\def\quadvar#1{\left<#1,#1\right>}

\begin{document}

\title{Volatility in options formulae for general stochastic dynamics}

\author{Kais Hamza}
\address{School of Mathematical Sciences, Building 28M, Monash
University, Clayton Campus, Victoria 3800, Australia.} \email{kais.hamza@monash.edu}

\author{Fima Klebaner}
\address{School of Mathematical Sciences, Building 28M, Monash
University, Clayton Campus, Victoria 3800, Australia.} \email{fima.klebaner@monash.edu}

\author{Olivia Mah}
\address{School of Mathematical Sciences, Building 28M, Monash
University, Clayton Campus, Victoria 3800, Australia.} \email{olivia.mah@monash.edu}
\thanks{Research supported by the Australian Research Council grant DP0988483.}



\maketitle

\begin{abstract}
It is well-known that the Black-Scholes formula has been derived under the assumption of constant volatility in stocks.
In spite of evidence that this parameter is not constant, this formula is widely used by financial markets.
This paper addresses the question of whether an alternative model for stock price exists for which the Black-Scholes
or similar formulae hold. The results obtained in this paper are very general as no assumptions are made on the
dynamics of the model, whether it be the underlying price process, the volatility process or how they relate to each
other. We show that if the formula holds for a continuum of strikes and three terminal times then the volatility must be
constant. However, when it only holds for finitely many strikes, and three or more maturity times, we obtain a
universal bound on the variation of the volatility. This bound yields that the implied volatility is constant
when the sequence of strikes  increases to cover the entire half-line.
This recovers the result for a continuum of strikes by a different approach.
\end{abstract}

\section{Introduction}

The goal of the paper is to examine the compatibility between the Black-Scholes formula
and stock price models with non-constant volatility. Our investigation has a much more general
setup than that of the Black-Scholes model encompassing a class of diffusion models for stock prices. 

The Black-Scholes option pricing formula (published 40 years ago) is based on the
Black-Scholes stock price model,
		\[
			dZ_t = \mu Z_t dt + \sigma Z_t dW_t,
		\]

where   $W_t$ is a Brownian motion and $\sigma$, known as the spot volatility, is assumed constant.
Here by the price of an option at time $t$ with strike $K$ under the model $Z_t$, $0\le t\le T$,
we understand the expression
$$C(T,t,K,\sigma,z) = e^{-r(T-t)}\mathbb{E}[(Z_T - K)^+ | Z_t = z],$$
where $r$ is the riskless interest rate and $T$ is the option maturity, $t\le T$.

The Black-Scholes {\it implied volatility} is defined as such value
of the volatility parameter $\sigma$ that, when plugged into the Black-Scholes formula $C(T,t,K,\sigma,z)$,
gives the observed market price, i.e. that value of $\sigma$ that makes the observed and the theoretical option price (or model price) coincide.

Empirical studies  have
shown that the Black-Scholes implied volatility varies with both strike prices $K$ and maturities $T$. This is called the smile effect, and contradicts the constant volatility assumption.
This fact, however, has not diminished the popularity of
the Black-Scholes formula.  Option prices are often quoted in terms of the Black-Scholes implied volatilities, making the Black-Scholes formula,
a convenient communication tool in the industry.

Finding a stock price model
which is compatible with the Black-Scholes formula is of tremendous interest
to the finance industry and has become the impetus for the development of the modeling of the Black-Scholes implied volatility.

In this paper our approach differs from that of other
implied volatility modeling work in one significant way: we do not place any assumptions on the general
dynamics of the stock price process, spot volatility process or how they are related.

Our research builds upon the work of Hamza and Klebaner \cite{HK2006}, where it was shown that if option prices
of an arbitrary stock price model are given by the Black-Scholes formula for
a continuum of strike prices $K$ (and three maturities), then the implied volatility
must be constant. Under the additional assumption that the filtration of the price model is that of a Brownian motion,
then the model must be the Black-Scholes model (see \cite{HK2007}). In \cite{HK2008} this
conclusion was extended to the case when the implied volatility is also assumed to depend on the maturity date $T$.

Here we extend this non-existence result in two directions. Firstly, we consider a class of diffusion models
much more general than the Black-Scholes model considered in \cite{HK2006} and reach the same non-existence
conclusion found therein. Secondly, and for the sake of practicality, we adopt the more realistic assumption
that the option price formula (Black-Scholes or otherwise) only holds for finitely many strikes.
In this case we arrive at the result that the implied volatility is not necessarily constant,
but approaches a constant as the number of strike prices increases. More importantly,
we show that the variation of the implied volatility process is bounded by a value that
depends only on the strike prices and not on any of the model parameters. Our main result
takes the form of a uniform bound on the variation of the implied volatility process, or
rather on a large family of proxies thereof, and also provides us with a set of constraints
limiting the acceptable (i.e. compatible with the option pricing formula) values of the stock
price and implied volatility parameters; the more maturity times we have, the more refined this
set of constraints would be.

The paper is organized as follows.  Section 2 provides the basic setup of the
paper.  In the following sections we state our results, first under the continuum of strikes assumption,
then under the finitely many strikes assumption.

\section{General setup}

Throughout this paper, $Z^{(\sigma)}_t$ will denote a reference process (eg the Black-Scholes model),
or rather a family of processes indexed by a parameter $\sigma$ (we call the volatility parameter),
while $S_t$ will stand for an unspecified process whose option price formula mimics
those of $Z^{(\sigma)}_t$ (a precise meaning is given later).
The two processes may live on separate spaces and we let $(\Omega,\mathcal{F},\mathbb{F},P)$
be the filtered probability space that supports $S_t$. As $Z^{(\sigma)}_t$ will be assumed Markovian,
any conditioning on the past reduces to a conditioning on the present, and therefore no reference
will be made to the filtration (or for that matter the probability space) of $Z^{(\sigma)}_t$.

For simplicity, we assume that the riskless interest rate $r=0$.
Therefore, according to the First Fundamental Theorem of Asset Pricing, a
stock price model does not have arbitrage opportunities if and
only if there exists an equivalent probability measure, known as a no-arbitrage measure,
under which the stock price is a martingale. The price of an option is then simply
the expectation, under the no-arbitrage measure, of the payoff function conditional on the past
\begin{equation}\label{E:gen_pricing}
    C(T,t,K,\sigma,z) = \mathbb{E}\Big[\Big(Z^{(\sigma)}_T - K\Big)^+ \Big| Z^{(\sigma)}_t = z\Big].
\end{equation}

Throughout this paper we will invariably write $\mathbb{E}$ for
the expectations in the ``$Z$-space'' or the ``$S$-space'';
the precise meaning of the expectation being
obvious from the context. For ease of exposition, we shall also often write $Z_t$ for
$Z^{(\sigma)}_t$, occasionally reverting to the original notation to highlight the dependence on $\sigma$.

As eluded earlier, the process $Z^{(\sigma)}_t$ will be assumed to be a diffusion whose diffusion
coefficient contains the volatility parameter $\sigma$ as a multiplicative factor:
\begin{equation}
dZ_t = \sigma h(t)\beta(Z_t)dB_t\label{sde}
\end{equation}
where $B$ is a Brownian motion, $h$ and $\beta$ are deterministic functions
such that $h\neq0$ and for any $t$, $\bds\int_0^th(s)^2ds<+\infty\eds$.

We   later comment on the reason for choosing a volatility function of the form $h(t)\beta(z)$ rather than
a more general $b(t,z)$.

Throughout this paper we make the following two assumptions.

{\bf Diffusion Assumption (D)}: Stochastic differential equation \eqref{sde}
has a unique weak solution that takes values in in an open interval in $[0,+\infty]$.

{\bf Martingale Assumption (M)}: $Z_t$ is a true martingale and there exists a deterministic function $\phi$
on $(0,+\infty)$ that is
\begin{enumerate}
\item[(M0)] positive,
\item[(M1)] of class $C^2$ and
\item[(M2)] such that, for any $\sigma>0$,
$\bds U_t=\exp\left(-\sigma^2\int_0^th(s)^2ds\right)\phi(Z_t)\eds$ is a true martingale
or equivalently that, for any $\sigma>0$, $\bds V_t=\phi(Z_t)-\sigma^2\int_0^th(s)^2\phi(Z_s)ds\eds$
is a true martingale.
\end{enumerate}

Observe that the smoothness assumption (M1) is almost redundant. Indeed, under the extra assumption that
$\quadvar Z_\infty=+\infty$, the martingale requirement in  (M2)
automatically implies that $\phi$ is the difference of 2 convex functions (a slightly weaker version of (M1)).
Indeed, since $\phi(Z_t)$ is a semimartingale, so is $\phi(Z_{\tau_t})$, where $\tau_t=\inf\{u:\quadvar Z_u>t\}$.
Now, $Z_{\tau_t}$ is a Brownian motion, and making use of Wang \cite{Wang1977}, we conclude that $\phi$ must be
the difference of 2 convex functions.

We also observe that the function $h$ plays no role in any of the martingale requirements in (M). Indeed, applying
the deterministic change of time $T_t=\inf\{s:\int_0^sh(u)^2du>t\}$ to the solution $Z$ yields a solution
of \eqref{sde} with $h\equiv1$.

\begin{example}
Of particular importance is, of course, the Black-Scholes model $dZ_t = \sigma Z_tdB_t$,
and its celebrated formula
$$C_{BS}(T,t,K,\sigma,z) = \mathbb{E}[(Z_T - K)^+ | Z_t = z]
= z\Phi(\eta)-K\Phi(\eta-\sigma\sqrt{T-t}),$$
where $\eta = \frac{\log\frac{z}{K}+\frac{\sigma^2}2(T-t)}{\sigma\sqrt{T-t}}$ and $\Phi$
is the distribution function of the standard normal distribution.
In this case
$h\equiv1$, $\beta(z)=z$ and, as demonstrated in \cite{HK2006}, one can choose $\phi(z)=z^2$.
\end{example}

While the choice of $\phi$ in the previous example may seem ``natural'' (see \cite{HK2006}),
this is not the case in general
and a generic way of finding a suitable $\phi$ would be desirable.
By application of Ito's formula, we see that a necessary condition we must impose on
$\phi$ is that it satisfies the ordinary differential equation
\begin{equation}\label{ode}
\frac12\beta(z)^2\phi''(z) = \phi(z).
\end{equation}
This condition on $\phi$ is necessary and sufficient to ensure that $U_t$ and $V_t$ are local martingales.
It is however not sufficient to guarantee that they are true martingales.

In the next example the martingale property of $U_t$ follows from its boundedness. Novikov and Kazamaki
conditions may be used in other cases -- see also Klebaner and Liptser \cite{KL2013}.

Note that, as a consequence of \eqref{ode}, any positive $\phi$ must also be (strictly) convex.

\begin{example}
If $Z_t$ is a scaled square of a 0-dimensional Bessel process, $dZ_t = \sigma\sqrt{Z_t}dB_t$, $Z_0=z_0>0$, then both
$Z_t$ and $U_t$ are true martingales. Recall that strong uniqueness holds in this case and that, by application of
Gronwall's and BDG inequalities, one can show that $Z_t$ is a square integrable martingale.

In this case, \eqref{ode} becomes $z\phi''(z) = 2\phi(z)$ and one can choose
$$\phi(z) = 2\sqrt{2z}K_1(2\sqrt{2z}),$$
where $K_1$ is the modified Bessel function of the second kind of order 1. Note that $\phi$ decreases from 1 to 0 and is
not differentiable at 0. It follows that $U$ is bounded but that its martingale property does not extend beyond
$\tau = \inf\{t\geq0:Z_t=0\}$.
\end{example}

\begin{example}
Consider the stochastic differential equation $dZ_t = \sigma Z_t\sqrt{-2\ln Z_t}dB_t$, $Z_0=z_0\in(0,1)$. Because
$\beta^{-2}$ is integrable on any compact subset of $(0,1)$, this SDE has a unique weak solution -- see
\cite[Theorem 5.15]{KS1991}.

Let $\tau = \inf\{t\geq0:Z_t=0\mbox{ or }Z_t=1\}$. Then the local martingale $Z_{t\wedge\tau}$ is bounded and is
therefore a true martingale.

Finally, $\phi(z)=-\ln z$ clearly solves \eqref{ode} and we have
\begin{eqnarray}
dU_t & = & \exp\left(-\sigma^2\int_0^th(s)^2ds\right)\phi'(Z_t)dZ_t\label{dU}\\
& = & -\sqrt2\sigma\exp\left(-\frac{\sigma^2}2\int_0^th(s)^2ds\right)\sqrt{U_t}dB_t.\nonumber
\end{eqnarray}
Thus we see that $U_t$ is, up to a deterministic change of time, nothing else but a squared
0-dimensional Bessel process. The true martingale property of $U_t$ immediately follows as  in the previous
example.
\end{example}

While in the examples above specific considerations enabled us to show the true martingale property of $Z_t$
and $U_t$ (and consequently that of $V_t$), it is natural to look for generic sufficient conditions to achieve
this requirement. In view of \eqref{dU}, the following result provides conditions for $U_t$ to be a true martingale
when $Z_t$ is known to be a true martingale.

\begin{theorem}
Let $\mathcal{Z}$ be a continuous martingale and $g$ be a convex function that satisfies the linear growth:
\begin{equation}
\exists\,\alpha>0\,/\ \forall z,\ |g(z)|\leq \alpha(1+|z|).
\end{equation}
Then the local martingale $\bds\int_0^tg'_-(\mathcal{Z}_s)d\mathcal{Z}_s\eds$ is a true martingale.
\end{theorem}

The crucial step in establishing this result is an extension of the Doob-Meyer decomposition for submartingales
to processes of class (DL) -- see \cite{HK2014} for details.

Delving further into the substance of Assumption (M), we see that in the time-homogeneous case, $h\equiv1$,
and denoting by $A$  the (extended) infinitesimal generator
of the process $Z^{(1)}$, the ``standard'' $Z$-process, \eqref{ode} can be simply
written as $A\phi=\phi$. In other words, (M) reduces to the existence of a deterministic function $\phi$
such that, for any $\sigma>0$,
$$P^{(\sigma)}_t\phi(z) = e^{\sigma^2t}\phi(z),$$
where $P^{(\sigma)}_t$ is the transition semigroup of $Z^{(\sigma)}$.

Finally, we comment on the choice of the product form adopted under Assumption (D)
rather than a more general class of models,
$dZ_t = \sigma b(t,Z_t)dB_t$. In this case the function $\phi$
in Assumption (M) ought to be allowed to depend on both time $t$
and $Z_t$. However, under sufficient smoothness and by Ito's formula we get
that $\phi$ must satisfy the identities
$$\frac{\partial}{\partial t}\phi(t,z) = 0\mbox{ and }
\frac12b(t,z)^2\frac{\partial^2}{\partial z^2}\phi(t,z) = h(t)^2\phi(t,z).$$
It now follows  that $\phi$ does not depend on $t$ and
$b(t,z)$ is of the product form given in (D).

In this paper, the family of processes $Z$ will be assumed to satisfy Assumptions (D) and (M) and
we aim to look for a pair of processes $(\theta_t,S_t)$ such that
\begin{equation}\label{maineq0}
\mathbb{E}[(S_T-K)^+|\mathcal{F}_t] = C(T,t,K,\theta_t, S_t),
\end{equation}
where the function $C$ is given by \eqref{E:gen_pricing}.
Initially we impose this condition for all $K$'s and three maturity times $T$. Later, we tackle
the case of finitely many strikes with a varying number of maturity times.

In what follows a  conditional expectation of the type $\mathbb{E}[g(Z_T)|Z_t=z]$ will be considered
a function of the triple $(t,\sigma,z)$. When such a function is applied to $(t,\theta_t,S_t)$,
we simply write $\mathbb{E}[g(Z_T)|Z_t=z](t,\theta_t,S_t)$.

\section{The case of a continuum of strikes}

Here we consider the case of a continuum of strikes and show that non-constant volatility parameter models are not
compatible with the option price formula \eqref{maineq0}. More specifically, when the option price formula holds for a continuum of strikes
(i.e. for any $K\geq0$) and three maturity times, then the process $\theta_t$ must  be constant (i.e.
non-random and not dependent on $t$). The following theorem is a generalization to diffusion models of the main result in
\cite{HK2006}.

\begin{theorem}\label{continuum}
Let $S_t$ and $\theta_t$ be adapted processes such that
$\theta_0 = \sigma$ and $S_0=z_0$. Assume that $S_t$ is
non-negative and that there exist
three terminal times, $T_1<T_2<T_3$ such that, for all
$K\geq0$ and all $t\leq T_i$, $i=1,2,3$
\begin{equation}\label{maineq}
\mathbb{E}[(S_{T_i}-K)^+|{\mathcal{F}}_t] =
C(T_i,t,K,\theta_t,S_t)
\end{equation}
Then $\theta^2_t = \sigma^2$ for all $t\le T_1$.

Furthermore, if $\mathbb{F}$ is the natural filtration of some Brownian motion,
then $(S_t)_{t\leq T_1}\stackrel{d}{=}(Z_t)_{t\leq T_1}$.
In other words $Z$ is the only model on a Brownian filtration that is compatible with \eqref{maineq}.
\end{theorem}

\begin{proof}
The proof follows \cite{HK2006} and is adapted to the general diffusion setup of Assumption (D)
through the use of the function $\phi$.

First we show that since $Z_t$ is a martingale so is $S_t$ (at least up to $T_3$). Indeed,
\begin{eqnarray*}
\mathbb{E}[S_{T_i}|{\mathcal{F}}_t] & = & \mathbb{E}[(S_{T_i}-0)^+|{\mathcal{F}}_t]\ =\ \mathbb{E}[(Z_{T_i}-0)^+|Z_t=z](t,\theta_t,S_t)\\
& = & \mathbb{E}[Z_{T_i}|Z_t=z](t,\theta_t,S_t)\ =\ S_t.
\end{eqnarray*}

Because
$$\phi(x) = \int_0^\infty(x - a)^{+} \phi''(a)da + \phi'(0)x + \phi(0),$$
we deduce that (recall that $\phi$ is convex and $\phi''\geq0$)
\begin{eqnarray*}
\lefteqn{\mathbb{E}[\phi(S_{T_i})|{\mathcal{F}}_t]}\\
& = & \int_0^\infty\mathbb{E}[(S_{T_i}-K)^+|{\mathcal{F}}_t]\phi''(K)dK + \phi'(0)\mathbb{E}[S_{T_i}|{\mathcal{F}}_t] + \phi(0)\\
& = & \int_0^\infty\mathbb{E}[(Z_{T_i}-K)^+|Z_t=z](t,\theta_t,S_t)\phi''(K)dK + \phi'(0)S_t + \phi(0)\\
& = & \mathbb{E}[\phi(Z_{T_i})|Z_t=z](t,\theta_t,S_t)\\
& = & \exp\left(\theta_t^2\int_t^{T_i}h(s)^2ds\right)\phi(S_t)
\end{eqnarray*}
where the last equality follows from Assumption (M2).
Thus for each $i$,\\
$\bds\exp\left(\theta_t^2\int_t^{T_i}h(s)^2ds\right)\phi(S_t)\eds$ is a
true martingale (up to $T_i$). In other words, with
$$M_t = \exp\left(\theta_t^2\int_t^{T_1}h(s)^2ds\right)\phi(S_t)\mbox{ and }X_t = \exp\left(\theta_t^2\int_{T_1}^{T_2}h(s)^2ds\right),$$
the assumptions of the theorem imply the existence of three martingales of the form $M_t$, $M_tX_t$ and $M_tX_t^\alpha$, where $X_t$
is a semimartingale and
$$\alpha = \frac{\int_{T_1}^{T_3}h(s)^2ds}{\int_{T_1}^{T_2}h(s)^2ds} > 1.$$
By Proposition \ref{specmartg} below it follows that $X_t=X_0$, which in turn proves that $\theta_t^2=\theta_0^2=\sigma^2$.

Now if $\mathbb{F}$ is a Brownian filtration, then all martingales are continuous and, in view of \eqref{ode} and the fact that $\theta_t^2=\sigma^2$,
\begin{eqnarray*}
\lefteqn{dM_t}\\
& = & \exp\left(\sigma^2\int_t^{T_1}h(s)^2ds\right)\left(-\sigma^2h(t)^2\phi(S_t)dt+\phi'(S_t)dS_t+\frac12\phi''(S_t)d\left<S,S\right>_t\right)\\
& = & \exp\left(\sigma^2\int_t^{T_1}h(s)^2ds\right)\left(-\frac{\sigma^2}2h(t)^2\beta(S_t)^2\phi''(S_t)dt+\frac12\phi''(S_t)d\left<S,S\right>_t\right)+d\bar{M}_t,
\end{eqnarray*}
where $\bar{M}_t$ is a (continuous) local martingale. As all local martingales of finite variation must be constant, it follows that
$$d\left<S,S\right>_t = \sigma^2h(t)^2\beta(S_t)^2dt.$$
\end{proof}

\begin{proposition}\label{specmartg}
Let $M_t$ and $X_t$ be two  semimartingales, $M_t$ is
positive. Assume that there exists $\alpha>1$ such that $M_t$, $M_tX_t$ and $M_tX_t^\alpha$ are local
martingales.
Then $X_t\equiv X_0$.
\end{proposition}

The proof of Proposition \ref{specmartg} is given in  \cite{HK2006} and is based on a convexity argument.

At this point, we stress that no assumptions are made on the dynamics of $S_t$ or on the
relationship between $S_t$ and $\theta_t$.

\section{The case of finitely many strikes}

We now only assume that \eqref{maineq} holds true for a finite sequence of strike prices, $0=K_0<K_1<K_2\ldots<K_{m}$.
With this reduced assumption, it is no longer possible to conclude that $\theta_t^2$ is constant. We shall see that
we can however obtain universal bounds on the amount of variation $\theta_t$ is allowed to have. These bounds are
universal in the sense that they do not depend on $S_t$ or $\theta_t$, but only on $Z_t$ or more specifically on
the function $\phi$.

The practical implication of these bounds, as given in Theorem \ref{mainth} and more specifically
inequality \eqref{bound}, is a set of constraints limiting the range of possible values
of the parameters that define $S_t$ and $\theta_t$. Indeed, for a given
parametric model, the left-hand side of \eqref{bound} can in principle be
expressed as a function of those parameters, while the right-hand side only
depends on $\phi$. Inequality \eqref{bound} can therefore be used as a tool in
calibrating the model parameters.

Similar to the case of continuum of strikes we extract information about the implied volatility process $\theta_t$ through the process
$X_t = \exp\left(\theta_t^2\int_{T_1}^{T_2}h(s)^2ds\right)$. In the former case we have shown that $X_t=X_0$, implying that
$\theta_t=\theta_0$, but in the present case we only manage a bound on the variation of $\theta$. The way we shall measure such a variation
is via a function $Q$ (a polynomial with possibly non-integer powers) associated with a sequence of terminal times, $T_1<T_2<\ldots<T_q$
($q\geq 3$). Let
$$\alpha_k = \frac{\int_{T_1}^{T_k}h(s)^2ds}{\int_{T_1}^{T_2}h(s)^2ds},\quad k=1,\ldots,q.$$
Observe that $\alpha_1=0$, $\alpha_2=1$ and the sequence $\alpha_k$ is increasing.
Now, for any non-zero sequence $p_3,\ldots,p_q$ of non-negative numbers, let
\begin{eqnarray}
p_2 & = & -\sum_{k=3}^qp_k\alpha_kX_0^{\alpha_k-1},\ p_1 = -\sum_{k=2}^qp_kX_0^{\alpha_k}\mbox{ and }\nonumber\\
Q(x) & = & \sum_{k=1}^qp_kx^{\alpha_k} = \sum_{k=3}^qp_k\left(x^{\alpha_k}-\alpha_kX_0^{\alpha_k-1}x+(\alpha_k-1)X_0^{\alpha_k}\right).\label{poly}
\end{eqnarray}
Such a function is clearly convex on $(0,\infty)$ and, nil and minimum at $x=X_0$. In particular
$Q(x)\geq0$.

In the particular time-homogeneous ($h\equiv1$) unweighted ($p_3=\ldots=p_q=1$) case of equidistant times ($T_{k+1}-T_k=T_2-T_1$), this function
reduces to $Q(x) = (x-X_0)^2P(x)$, for some polynomial $P$ with positive coefficients.

Let, for a terminal time $T$ and $t\leq T$,
\begin{eqnarray*}
M_{t,T} & = & \mathbb{E}[\phi(S_T)|{\mathcal{F}}_t]\mbox{ and}\\
N_{t,T} & = & \mathbb{E}[\phi(Z_T)|Z_t=z](t,\theta_t,S_{t})\ =\ \exp\left(\theta_t^2\int_t^Th(s)^2ds\right)\phi(S_t).
\end{eqnarray*}

In the case of a continuum of strikes, we saw that the martingale $M_{t,T}$ coincides with $N_{t,T}$. This fact
plays a key role in establishing the non-existence result in continuum of strikes case. This non-existence result is no longer true
in the case of finitely many strikes. Therefore our initial aim is precisely to compute the difference
$$D_{t,T} = M_{t,T}-N_{t,T}.$$

\begin{proposition}
Let $S_t$ be a non-negative adapted process. Assume that there exist a finite sequence of strike prices,
$0=K_0<K_1<K_2\ldots<K_{m}$ such that, for all $t\leq T$,
\begin{eqnarray*}
\lefteqn{\mathbb{E}[(S_T-K_j)^+|{\mathcal{F}}_t]}\\
& = & C(T,t,K_j,\theta_t,S_t)\ =\ \mathbb{E}[(Z_T-K_j)^+|Z_t=z](t,\theta_t,S_{t})
\end{eqnarray*}
Then
\begin{eqnarray}
\lefteqn{D_{t,T}}\nonumber\\
& = & \sum_{j=0}^{m-1}\int_{K_j}^{K_{j+1}}\Big(\mathbb{E}[(S_T-K)^+|{\mathcal{F}}_t]
-\mathbb{E}[(S_T-K_j)^+|{\mathcal{F}}_t]\Big)\phi''(K)dK\nonumber\\
& & -\ \sum_{j=0}^{m-1}\int_{K_j}^{K_{j+1}}\Big(\mathbb{E}[(Z_T-K)^+|Z_t=z](t,\theta_t,S_t)\label{DtT}\\
& & -\ \mathbb{E}[(Z_T-K_j)^+|Z_t=z](t,\theta_t,S_t)\Big)\phi''(K)dK\nonumber\\
& & +\ \mathbb{E}[\hat\phi(K_m,S_T)|{\mathcal{F}}_t]-\mathbb{E}[\hat\phi(K_m,Z_T)|Z_t=z](t,\theta_t,S_t),\nonumber
\end{eqnarray}
where $\hat\phi(b,x) = \phi(x)-\phi(x\wedge b) = (\phi(x)-\phi(b))1_{x>b}$.
\end{proposition}

\begin{proof}
Since $\mathbb{E}[(S_T-0)^+|{\mathcal{F}}_t] = \mathbb{E}[(Z_T-0)^+|Z_t=z](t,\theta_t,S_{t})$,
\begin{eqnarray*}
\lefteqn{D_{t,T}}\\
& = & \int_0^\infty\mathbb{E}[(S_T-K)^+|{\mathcal{F}}_t]\phi''(K)dK + \phi'(0)\mathbb{E}[(S_T-0)^+|{\mathcal{F}}_t] + \phi(0)\\
& & -\ \int_0^\infty\mathbb{E}[(Z_T-K)^+|Z_t=z](t,\theta_t,S_{t})\phi''(K)dK\\
& & -\ \phi'(0)\mathbb{E}[(Z_T-0)^+|Z_t=z](t,\theta_t,S_{t}) - \phi(0)\\
& = & \int_0^\infty\mathbb{E}[(S_T-K)^+|{\mathcal{F}}_t]\phi''(K)dK\\
& & -\ \int_0^\infty\mathbb{E}[(Z_T-K)^+|Z_t=z](t,\theta_t,S_{t})\phi''(K)dK
\end{eqnarray*}
Therefore, using the identity
$$\mathbb{E}[(S_T-K_j)^+|{\mathcal{F}}_t] = \mathbb{E}[(Z_T-K_j)^+|Z_t=z](t,\theta_t,S_{t}),$$
we get
\begin{eqnarray*}
\lefteqn{D_{t,T}}\\
& = & \sum_{j=0}^{m-1}\int_{K_j}^{K_{j+1}}\mathbb{E}[(S_T-K)^+|{\mathcal{F}}_t] + \int_{K_m}^\infty\mathbb{E}[(S_T-K)^+|{\mathcal{F}}_t]\phi''(K)dK\\
& & -\ \sum_{j=0}^{m-1}\int_{K_j}^{K_{j+1}}\mathbb{E}[(Z_T-K)^+|Z_t=z](t,\theta_t,S_t)\phi''(K)dK\\
& & -\  \int_{K_m}^\infty\mathbb{E}[(Z_T-K)^+|Z_t=z](t,\theta_t,S_t)\phi''(K)dK
\end{eqnarray*}
that is,
\begin{eqnarray*}
\lefteqn{D_{t,T}}\\
& = & \sum_{j=0}^{m-1}\int_{K_j}^{K_{j+1}}\Big(\mathbb{E}[(S_T-K)^+|{\mathcal{F}}_t]
-\mathbb{E}[(S_T-K_j)^+|{\mathcal{F}}_t]\Big)\phi''(K)dK\\
& & -\ \sum_{j=0}^{m-1}\int_{K_j}^{K_{j+1}}\Big(\mathbb{E}[(Z_T-K)^+|Z_t=z](t,\theta_t,S_t)\\
& & -\ \mathbb{E}[(Z_T-K_j)^+|Z_t=z](t,\theta_t,S_t)\Big)\phi''(K)dK\\
& & +\ \mathbb{E}[\hat\phi(K_m,S_T)|{\mathcal{F}}_t]-\mathbb{E}[\hat\phi(K_m,Z_T)|Z_t=z](t,\theta_t,S_t),
\end{eqnarray*}
which completes the proof.
\end{proof}

We are now ready to state the main result of this section. We shall rely on two facts:
\begin{enumerate}
\item $N_{t,T} + D_{t,T}$ is a true martingale;
\item for a sequence of maturity times $T_i$, $N_{t,T_k}=N_{t,T_1}X_t^{\alpha_k}$, for some $\alpha_k$,
and where as before
\begin{equation}\label{Xt}
X_t = \exp\left(\theta_t^2\int_{T_1}^{T_2}h(s)^2ds\right).
\end{equation}
\end{enumerate}

\begin{theorem}\label{mainth}
Let $S_t$ and $\theta_t$ be adapted processes such that
$\theta_0 = \sigma$ and $S_0=z_0$. Assume that $S_t$ is
non-negative and that there exist a sequence of terminal times, $T_1<T_2<\ldots<T_q$ and a finite sequence of strike prices,
$0=K_0<K_1<K_2\ldots<K_{m}$ such that, for all $t\leq T_i$,
\begin{equation}\label{maineqfinite}
\mathbb{E}[(S_{T_i}-K_j)^+|{\mathcal{F}}_t] =
C(T_i,t,K_j,\theta_t,S_t)
\end{equation}
Then for any non-zero sequence $p_3,\ldots,p_q$ of non-negative numbers, or equivalently for any function $Q$ as defined in \eqref{poly},
and assuming that $\mathbb{E}[N_{t,T_q}]<+\infty$,
\begin{eqnarray}
\lefteqn{\left|\mathbb{E}[N_{t,T_1}Q(X_t)]+\sum_{k=1}^qp_k\mathbb{E}[G_{0,T_k}-G_{t,T_k}]\right|}\nonumber\\
& \leq & 2\sum_{k=1}^q|p_k|\sum_{j=0}^{m-1}(K_{j+1}-K_j)(\phi'(K_{j+1})-\phi'(K_j)),\label{bound}
\end{eqnarray}
where $X_t$ is given in \eqref{Xt} and
$$G_{t,T} = \mathbb{E}[\hat{\phi}(K_m, Z_T)|Z_t=z](t,\theta_t,S_t).$$
\end{theorem}

\begin{proof}
Since $N_{t,T_k}=N_{t,T_1}X_t^{\alpha_k}$ and $N_{t,T} + D_{t,T} = M_{t,T}$,
$$N_{t,T_1}Q(X_t)+\sum_{k=1}^qp_kD_{t,T_k} = \sum_{k=1}^qp_k(N_{t,T_k}+D_{t,T_k}) = \sum_{k=1}^qp_kM_{t,T_k}.$$
Further, labeling the various terms in \eqref{DtT},
\begin{eqnarray*}
H_{t,T}	& = & \sum_{j=0}^{m-1}\int_{K_j}^{K_{j+1}}\left(\mathbb{E}[(S_T-K)^+|\mathcal{F}_t]-\mathbb{E}[(S_T-K_j)^+|\mathcal{F}_t]\right)  \phi''(K)dK \\
& &	+\ \mathbb{E}[\hat{\phi}(K_m, S_T)|\mathcal{F}_t]\\
L_{t,T} & = & \sum_{j=0}^{m-1}\int_{K_j}^{K_{j+1}}\Big(\mathbb{E}[(Z_T-K)^+|Z_t=z](t,\theta_t,S_t)\\
& &	-\ \mathbb{E}[(Z_T-K_j)^+|Z_t=z](t,\theta_t,S_t)\Big)\phi''(K)dK\\
G_{t,T} & = & \mathbb{E}[\hat{\phi}(K_m, Z_T)|Z_t=z](t,\theta_t,S_t),
\end{eqnarray*}
so that $D_{t,T} = H_{t,T} - L_{t,T} - G_{t,T}$, and noting that $Q(X_0)=0$ and that both $M_{t,T}$ and $H_{t,T}$ are martingales
(see \cite{HJK2005} for a proof of the latter), we get,
$$\mathbb{E}[N_{t,T_1}Q(X_t)]+\sum_{k=1}^qp_k\mathbb{E}[D_{t,T_k}] = \sum_{k=1}^qp_k\mathbb{E}[M_{0,T_k}]
= \sum_{k=1}^qp_k\mathbb{E}[D_{0,T_k}],$$
from which it follows that
$$\mathbb{E}[N_{t,T_1}Q(X_t)] + \sum_{k=1}^qp_k\mathbb{E}[G_{0,T_k}-G_{t,T_k}] = \sum_{k=1}^qp_k\mathbb{E}[L_{t,T_k}-L_{0,T_k}].$$
Now, noting that, for $K\in(K_j,K_{j+1})$,
$$-(K-K_j)\leq(z-K)^+-(z-K_j)^+\leq0,$$
we get that
$$-\sum_{j=0}^{m-1}(K_{j+1}-K_j)(\phi'(K_{j+1})-\phi'(K_j))\leq L_{t,T}\leq 0.$$
Recall that since $\phi$ is convex, $\phi'(K_{j+1})-\phi'(K_j)\geq0$.
The result follows immediately.
\end{proof}

We note that, since $\mathbb{E}[\hat{\phi}(K_m, Z_T)|Z_t=z]$ converges to 0 as $K_m\to\infty$,
by taking a large enough $K_m$, we can make the term $\sum_{k=1}^qp_k\mathbb{E}[G_{0,T_k}-G_{t,T_k}]$ arbitrarily small so that
the left-hand side of \eqref{bound} essentially reduces to $\mathbb{E}[N_{t,T_1}Q(X_t)]$, which
in the time-homogeneous ($h\equiv1$) unweighted ($p_3=\ldots=p_q=1$) case of three equidistant times ($T_3-T_2=T_2-T_1$),
further reduces to $\mathbb{E}[N_{t,T_1}(X_t-X_0)^2]$.

We stress here that we have complete freedom as to
the choice of the non-zero sequence of non-negative numbers $p_3,\ldots,p_q$,
thus providing us with an infinite set of constraints.

Next we show that if the number of strike prices
increases in the sense of \eqref{strikes},
then the bound in~\eqref{bound} vanishes.  In other words,
the implied volatility approaches a constant when the number of
strike prices increases.  The notations introduced
in Theorem \ref{mainth} will now carry an index $n$.

\begin{corollary}\label{corconv}
Let $T_1<T_2<\ldots<T_q$ be a sequence of terminal times.
Let, for each $n$, $S^{(n)}_t$ and $\theta^{(n)}_t$ be adapted processes such that
$\theta^{(n)}_0 = \sigma$ and $S^{(n)}_0=z_0$. Assume that $S^{(n)}_t$ is
non-negative and that there exist a finite sequence of strike prices,
$0=K^{(n)}_0<K^{(n)}_1<K^{(n)}_2\ldots<K^{(n)}_{m(n)}$ such that, $\lim_{n\to\infty}K^{(n)}_{m(n)}=+\infty$,
\begin{equation}\label{strikes}
\lim_{n\to\infty} K^{(n)}_{m(n)} \,\max_{0\leq j\leq m(n) - 1}\big(\phi'(K^{(n)}_{j+1})-\phi'(K^{(n)}_j)\big) = 0
\end{equation}
and \eqref{maineqfinite} holds true for all $n$, all $j\leq m(n)$ and all $t\leq T_i$,
then for all $t\leq T_1$,
$$\liminf_{n\to\infty}\theta^{(n)}_t = \sigma.$$
\end{corollary}

\begin{proof}
First we observe that for any $z$ and $n$ large enough, $\hat\phi(K^{(n)}_{m(n)},z)=0$ and that
$\hat\phi(K^{(n)}_{m(n)},Z_T)$ is dominated by the integrable random variable $\phi(Z_T)$. It follows,
by dominated convergence, that $\mathbb{E}[G^{(n)}_{t,T_k}]$ approaches 0. Since the right-hand side of \eqref{bound}
clearly approaches 0, we immediately get that
$$\lim_{n\to\infty}\mathbb{E}[N^{(n)}_{t,T_1}Q(X^{(n)}_{t})] =0.$$
Since $N^{(n)}_{t,T_1}Q(X^{(n)}_{t})\geq0$, we deduce by Fatou's Lemma that $\liminf_nQ(X^{(n)}_{t})=0$. We complete the proof
by recalling that $Q$ is convex and nil at $X_0$.

Here a superscript $(n)$ has been added whenever $(S_t,\theta_t)$ is replaced with $(S^{(n)}_t,\theta^{(n)}_t)$.
\end{proof}

\begin{corollary}\label{thconv}
Let $T_1<T_2<\ldots<T_q$ be a sequence of terminal times.
Let, for each $n$, $S^{(n)}_t$ and $\theta^{(n)}_t$ be adapted processes such that
$\theta^{(n)}_0 = \sigma$ and $S^{(n)}_0=z_0$. Assume that $S^{(n)}_t$ is
non-negative and that the sequence of processes $\theta^{(n)}_t$
is tight. Further assume that there exist a finite sequence of strike prices,
$0=K^{(n)}_0<K^{(n)}_1<K^{(n)}_2\ldots<K^{(n)}_{m(n)}$ such that, $\lim_{n\to\infty}K^{(n)}_{m(n)}=+\infty$,
\begin{equation}\label{strikes}
\lim_{n\to\infty} K^{(n)}_{m(n)} \,\max_{0\leq j\leq m(n) - 1}\big(\phi'(K^{(n)}_{j+1})-\phi'(K^{(n)}_j)\big) = 0
\end{equation}
and \eqref{maineqfinite} holds true for all $n$, all $j\leq m(n)$ and all $t\leq T_i$.
Then for all $t\leq T_1$,
$$\lim_{n\to\infty}\theta^{(n)}_t = \sigma.$$
\end{corollary}

\begin{proof}
The result follows from Corollary \ref{corconv} above.
By Prohorov's Theorem, the sequence of processes $\theta^{(n)}_t$ is sequentially compact.
As a result, applying Corollary \ref{corconv} to any convergent subsequence yields
$$\lim_{k\to\infty}\theta^{(n_k)}_t = \liminf_{k\to\infty}\theta^{(n_k)}_t = \sigma.$$
As the limit is the same for all subsequences, the sequence itself must converge, to $\sigma$.
\end{proof}

In particular, if $\theta^{(n)}_t = \theta_t$ (does not depend on $n$), then Theorem \ref{thconv} states that $\theta_t$ is the constant $\sigma$.
Finally, the tightness assumption on the sequence of processes $\theta^{(n)}_t$ can be verified by checking Aldous' condition (see \cite{Aldous1978}).

\end{document}